\documentclass[11pt]{article}% добавляется twoside для двусторонней печати

\usepackage[T1]{fontenc}
\usepackage[cp1251]{inputenc}
\usepackage{textcomp}
\usepackage[centertags]{amsmath}
\usepackage[mediummath]{nccmath}
\usepackage{amsfonts}
\usepackage{amssymb}
\usepackage[hypertex]{hyperref}% не работает у меня
\usepackage{graphicx}
\usepackage[numbers,sort&compress]{natbib}% упорядочивает ссылки

\usepackage{amsthm}%Окружение для доказательств.

\usepackage{paperinitial}% Установка параметров страницы

\paperinitialization{15mm}{15mm}{15mm}{15mm}{2pt}{10pt}

\DeclareMathOperator{\im}{Im}

\newcommand{\vf}{\varphi}
\newcommand{\vk}{\varkappa}

\newcommand{\de}{\delta}

\newcommand{\la}{\lambda}
\newcommand{\La}{\Lambda}

\newtheorem{proposition}{Proposition}
\newtheorem{theorem}{Theorem}[section]
\newtheorem*{remark}{Remark}

\begin{document}
\allowdisplaybreaks[4]% позволяет переносить многострочные формулы
\frenchspacing% уменьшение пробелов после запятых
\setlength{\unitlength}{1pt}% устанавливает единицу длины в окружении picture

\title{{\Large\textbf{Multichannel scattering for the Schr\"{o}dinger \\ equation on a line with different thresholds at both infinities}}}

\date{}

\author{P.O. Kazinski\thanks{E-mail: \texttt{kpo@phys.tsu.ru}}\;\, and P.S. Korolev\thanks{E-mail: \texttt{kizorph.d@gmail.com}}\\[0.5em]
{\normalsize Physics Faculty, Tomsk State University, Tomsk 634050, Russia}}

\maketitle

\begin{abstract}

The multichannel scattering problem for the stationary Schr\"{o}dinger equation on a line with different thresholds at both infinities is investigated. The analytical structure of the Jost solutions and of the transition matrix relating the Jost solutions as functions of the spectral parameter is described. Unitarity of the scattering matrix is proved in the general case when some of the scattering channels can be closed and the thresholds can be different at left and right infinities on the line. The symmetry relations of the $S$-matrix are established. The condition determining the bound states is obtained. The asymptotics of the Jost functions and of the transition matrix are derived for a large spectral parameter.

\end{abstract}

\section{Introduction}

The scattering problem for a one-dimensional matrix Schr\"{o}dinger equation is a classical problem of quantum theory. The exhaustive treatment of general properties of such scattering on semiaxis, in particular, the proof of unitarity of the $S$-matrix in the presence of closed scattering channels, is given in the book \cite{Newton_scat}, Chap. 17. As far as multichannel scattering on the whole line is concerned, this problem was investigated in many works especially in regard to the inverse scattering problem and the construction of exact solutions to the hierarchies of integrable nonlinear partial differential equations \cite{Zakharov1980}. Nevertheless, to our knowledge, the description of properties of the $S$-matrix, of the Jost solutions, and of the bound states in the general case of multichannel scattering on a line with different thresholds at both left and right infinities is absent in the literature. Our aim is to fill this gap.

The study of the analytical structure of the Jost solutions and the $S$-matrix for one-channel scattering on the whole line in relation to the inverse scattering problem was being conducted already at the end of the 50s \cite{Faddeev1958,Kay1960,Faddeev1964,Faddeev1974}. As for the relatively recent papers regarding one-channel scattering including scattering on potentials with asymmetric asymptotics at left and right infinities, see, e.g., \cite{Newton1980,Zakharov1980,Aktosun1992,Aktosun1994,Gesztesy1997,Ostrovsky2005,deMonvel2008}. The general properties of two-channel scattering for both identical and distinct thresholds were considered in \cite{Ning1995,melgaard2001,melgaard2002,vanDijk2008}, the thresholds being the same at both infinities. In the papers \cite{Wadati1974,Calogero1976,Wadati1980,Alonso1982,Olmedilla1985,Zakhariev1990,Kiers1996,Aktosun2001,Corona2004,Bondarenko2017}, these results were generalized to multichannel scattering on a line for the case when all the reaction thresholds coincide at both left and right infinities. The multichannel scattering problem with different thresholds identical at both infinities was investigated in \cite{Sofianos1997,Braun2003}. In the works \cite{Zakharov1980,Aktosun2000,Sakhnovich2003}, the multichannel scattering problem for systems of a Hamiltonian type with matrix potentials vanishing at both infinities was studied. However, unitarity of the $S$-matrix in the presence of closed channels was not proved in these papers. In the present paper, we prove unitarity of the scattering matrix for a stationary matrix Schr\"{o}dinger equation on a line in the general case where the reaction thresholds do not coincide at both infinities and some of the channels are closed. Furthermore, we obtain the other relations connecting the transmission and reflection matrices for open and closed channels.

For such a scattering problem, we describe the analytical structure of the Jost solutions and of the transition matrix relating the bases of the Jost solutions. We prove the necessary and sufficient condition specifying the positions of the bound states. The form of this condition is well-known for multichannel scattering (see, e.g., \cite{Zakharov1980,Wadati1974,Wadati1980,Bondarenko2017}). However, we show that this condition also holds for the scattering problem with different thresholds at both left and right infinities.

The paper is organized as follows. In Sec. \ref{Jost_Sols_Sec}, the analytical properties of the Jost solutions are described. Sec. \ref{Analyt_Prop_Tran_Mat_Sec} is devoted to analytical properties of the transition matrix. In Sec. \ref{Basic_Ident_Sec}, the basic identities for the scattering matrix are discussed. In Sec. \ref{Identities_Open_Chan_Sec}, the relations between the transmission and reflection matrices in open channels are investigated. Sec. \ref{Unitarity_Sec} is devoted to the proof of unitarity of the $S$-matrix in open channels. In Sec. \ref{Bouns_States_Sec}, we discuss the necessary and sufficient condition determining the location of bound states. In Sec. \ref{Shartwave_Asypt_Sec}, we obtain the asymptotics of the Jost solutions and of the transition matrix for a large spectral parameter. In Conclusion, we summarize the results. The summation over repeated indices is always understood unless otherwise stated. Furthermore, wherever it does not lead to misunderstanding, we use the matrix notation.

%\selectlanguage{english}
%\newpage
\section{Jost solutions and their analytical properties}\label{Jost_Sols_Sec}

Consider the matrix ordinary differential equation
\begin{equation}\label{Schr_eqn}
	\big[\partial_zg_{ij}(z)\partial_z +V_{ij}(z;\la)\big]u_j(z)=0,\qquad z\in \mathbb{R},\quad i,j=\overline{1,N},
\end{equation}
where $\la\in \mathbb{C}$ is an auxiliary parameter,
\begin{equation}
	V_{ij}(z;\la)=V_{ij}(z)-\la g_{ij}(z),
\end{equation}
and the matrices $g_{ij}(z)$ and $V_{ij}(z)$ are real and symmetric. The elements of these matrices are piecewise continuous functions. The matrix $g_{ij}(z)$ is positive-definite. We also assume that there exists $L_z>0$ such that
\begin{equation}\label{finite_supp}
	g_{ij}(z)|_{z>L_z}=g^+_{ij},\qquad V_{ij}(z)|_{z>L_z}=V^+_{ij},\qquad g_{ij}(z)|_{z<-L_z}=g^-_{ij},
    \qquad V_{ij}(z)|_{z<-L_z}=V^-_{ij},
\end{equation}
where $g^\pm_{ij}$ and $V^\pm_{ij}$ are constant matrices. The spectral parameter $\la$ is not the energy, in general. For example, the energy enters into the matrices $g_{ij}(z)$ and $V_{ij}(z)$ as a parameter for the scattering problems in electrodynamics of dispersive media and the physical value of $\la$ is zero in this case. The corresponding nonstationary scattering problem is not described by the nonstationary Schr\"{o}dinger equation associated with \eqref{Schr_eqn}. Notice that all the results of the present paper are applicable to the case where the matrices $g_{ij}(z)$ and $V_{ij}(z)$ are not real and symmetric but Hermitian. In that case, one just has to separate the real and imaginary parts of the initial matrix Schr\"{o}dinger equation. The resulting system of equations will be of the form \eqref{Schr_eqn} but of twice the size of the initial system.

Let
\begin{equation}
	g^\pm_{ij}f^\pm_{js}\La^\pm_s=V^\pm_{ij}f^\pm_{js} \qquad\text{(no summation over $s$)},\qquad s=\overline{1,N},
\end{equation}
where $\La^\pm_s\in \mathbb{R}$ are eigenvalues and the following normalization condition is true
\begin{equation}
	f_{\pm}^Tg_\pm f_\pm=1.
\end{equation}
We consider the general case where all of $\La^+_s$ and all of $\La^-_s$ are different. The degenerate
case is obtained by going to the respective limit. Let us introduce the diagonal matrices
\begin{equation}
	K^\pm_{ss'}:=\sqrt{\La^\pm_s-\la}\de_{ss'},
\end{equation}
where the principal branch of the square root is chosen. In particular, if  $\la>\La^\pm_s$, then $\sqrt{\La^\pm_s-\la}=i\sqrt{\la-\La^\pm_s}$. By definition, the Jost solutions to Eq. \eqref{Schr_eqn} have the asymptotics
\begin{equation}\label{Jost_asympt}
	(F^+_\pm)_{is}(z;\la)\underset{z\rightarrow\infty}{\rightarrow}(f_+)_{is'}(e^{\pm iK_+z})_{s's},
    \qquad (F^-_\pm)_{is}(z;\la)\underset{z\rightarrow-\infty}{\rightarrow}(f_-)_{is'}(e^{\pm iK_-z})_{s's}.
\end{equation}
For these solutions, we can write
\begin{equation}\label{Jost_sols}
\begin{split}
	F^+_+(z;\la)&=f_+e^{iK_+z}+\int_z^\infty dtf_+\frac{\sin K_+(z-t)}{K_+}f^T_+U_+(t;\la)F^+_+(t;\la),\\
	F^+_-(z;\la)&=f_+e^{-iK_+z}+\int_z^\infty dtf_+\frac{\sin K_+(z-t)}{K_+}f^T_+U_+(t;\la)F^+_-(t;\la),\\
	F^-_+(z;\la)&=f_-e^{iK_-z}-\int_{-\infty}^z dtf_-\frac{\sin K_-(z-t)}{K_-}f^T_-U_-(t;\la)F^-_+(t;\la),\\
	F^-_-(z;\la)&=f_-e^{-iK_-z}-\int_{-\infty}^z dtf_-\frac{\sin K_-(z-t)}{K_-}f^T_-U_-(t;\la)F^-_-(t;\la),
\end{split}
\end{equation}
where
\begin{equation}
	U_\pm(z;\la):=\partial_zg (z)\partial_z+V(z;\la)-\partial_zg_\pm\partial_z-V_\pm+\la g_\pm.
\end{equation}
In virtue of the assumption \eqref{finite_supp}, the integration in the integral representations of the Jost solutions is performed over a finite interval. Therefore, the Jost solutions are analytic functions of $\la$ on a double-sheeted Riemann surface. The solutions $(F^+_\pm)_{is}$ are the different branches of the same vector-valued analytic function of $\la$ with the branching point $\la=\La^+_s$, whereas the solutions $(F^-_\pm)_{is}$ are the different branches of the same vector-valued analytic function of $\la$ with the branching point $\la=\La^-_s$.

\section{Analytical properties of the transition matrix}\label{Analyt_Prop_Tran_Mat_Sec}

The Jost solutions $F^+_\pm$ and $F^-_\pm$ constitute bases in the space of solutions of Eq. \eqref{Schr_eqn}. Consequently,
\begin{equation}\label{basis_expan}
	F^+_+=F^-_+\Phi_++F^-_-\Psi_+,\qquad F^+_-=F^-_+\Psi_- +F^-_-\Phi_-,
\end{equation}
where $(\Phi_\pm)_{ss'}(\la)$ and $(\Psi_\pm)_{ss'}(\la)$ are some $z$-independent matrices. It is clear, that the Wronskian,
\begin{equation}
	w[\vf,\psi]:=\vf^T(z)g(z)\partial_z\psi(z)-\partial_z\vf^T(z)g(z)\psi(z),
\end{equation}
of two solutions $\vf(z)$ and $\psi(z)$ of Eq. \eqref{Schr_eqn} is independent of $z$ and defines a skew-symmetric scalar product on the space of solutions of Eq. \eqref{Schr_eqn}. From the asymptotics \eqref{Jost_asympt} we have
\begin{equation}\label{basic_wronsk}
	w[F^\pm_+,F^\pm_+]=w[F^\pm_-,F^\pm_-]=0,\qquad w[F^\pm_+,F^\pm_-]=-2iK_\pm.
\end{equation}
This skew-symmetric scalar product allows one to express the matrices $\Phi_\pm$ and $\Psi_\pm$ in terms of Wronskians of the Jost solutions
\begin{equation}\label{PhiPsiWronsk}
	\begin{aligned}
		2iK_-\Phi_+&=w[F^-_-,F^+_+],&\qquad -2iK_-\Phi_-&=w[F^-_+,F^+_-],\\
		-2iK_-\Psi_+&=w[F^-_+,F^+_+],&\qquad 2iK_-\Psi_-&=w[F^-_-,F^+_-].
	\end{aligned}
\end{equation}
We see from these relations that $(\Phi_\pm)_{ss'}$ and $(\Psi_\pm)_{ss'}$ are analytic functions of $\la$ with branching points of the square root type at $\la=\La^-_s$ and $\la=\La^+_{s'}$. The four functions $(\Phi_\pm)_{ss'}$, $(\Psi_\pm)_{ss'}$ are the four branches of the same analytic function of $\la$. Indeed, bypassing the branching point $\la=\La^-_s$, we have
\begin{equation}
	(K_-)_s\rightarrow-(K_-)_s,\qquad (F^-_\pm)_{is}\rightarrow (F^-_\mp)_{is}.
\end{equation}
Then, using \eqref{PhiPsiWronsk}, we obtain
\begin{equation}
	(\Phi_\pm)_{ss'}\rightarrow(\Psi_\pm)_{ss'},\qquad (\Psi_\pm)_{ss'}\rightarrow(\Phi_\pm)_{ss'}.
\end{equation}
Similarly, bypassing the branching point $\la=\La^+_{s'}$, we come to
\begin{equation}
	(\Phi_\pm)_{ss'}\rightarrow(\Psi_\mp)_{ss'},\qquad (\Psi_\pm)_{ss'}\rightarrow(\Phi_\mp)_{ss'}.
\end{equation}
Hence, starting from $(\Phi_+)_{ss'}$ and bypassing successively the branching points $\la=\La^-_{s}$ and $\la=\La^+_{s'}$, we obtain all the four functions $(\Phi_\pm)_{ss'}$, $(\Psi_\pm)_{ss'}$.

Consider the action of complex conjugation on the matrices $\Phi_\pm(\la)$ and $\Psi_\pm(\la)$. If $\la$ does not belong to the cuts of the functions $(K_-)_s$ and $(K_+)_{s'}$, then, using \eqref{Jost_sols} and \eqref{PhiPsiWronsk}, we obtain the following relations
\begin{equation}\label{PhiPsi_cc_regul}
	(\Phi_\pm)^*_{ss'}(\la)=(\Phi_\mp)_{ss'}(\la^*),\qquad (\Psi_\pm)^*_{ss'}(\la)=(\Psi_\mp)_{ss'}(\la^*).
\end{equation}
If $\la$ lies on the cut of the function $(K_-)_s$ but does not belong to the cut of the function $(K_+)_{s'}$, then
\begin{equation}
	(\Phi_\pm)^*_{ss'}(\la)=(\Psi_\mp)_{ss'}(\la).
\end{equation}
If $\la$ lies on the cut of the function $(K_+)_{s'}$ but does not belong to the cut of function  $(K_-)_{s}$, then
\begin{equation}
	(\Phi_\pm)^*_{ss'}(\la)=(\Psi_\pm)_{ss'}(\la).
\end{equation}
If $\la$ lies on the cuts of the functions $(K_-)_{s}$ and $(K_+)_{s'}$, we have
\begin{equation}
	(\Phi_\pm)^*_{ss'}(\la)=(\Phi_\pm)_{ss'}(\la),\qquad (\Psi_\pm)^*_{ss'}(\la)=(\Psi_\pm)_{ss'}(\la),
\end{equation}
i.e., in this case $(\Phi_\pm)_{ss'}(\la)$ and $(\Psi_\pm)_{ss'}(\la)$ are real.

\section{Basic identities for the scattering matrix}\label{Basic_Ident_Sec}

Formulas \eqref{basis_expan}, \eqref{basic_wronsk} yield the relations
\begin{equation}\label{PhiPsi_rels}
\begin{aligned}
	\Phi^T_+ K_-\Psi_+-\Psi^T_+K_-\Phi_+&=0,&\qquad \Phi^T_+ K_-\Phi_- -\Psi^T_+K_-\Psi_-&=K_+,&\qquad \Phi^T_- K_-\Psi_-
    -\Psi^T_-K_-\Phi_-&=0,\\
	\Phi_+ K_+^{-1} \Psi_-^T -\Psi_-K_+^{-1} \Phi_+^T&=0,&\qquad \Phi_+ K_+^{-1} \Phi_-^T -\Psi_-K_+^{-1} \Psi_+^T&=K_-^{-1},&\qquad
    \Phi_- K_+^{-1}\Psi_+^T -\Psi_+K_+^{-1}\Phi_-^T&=0.
\end{aligned}
\end{equation}
Let us introduce the transmission matrices $t_{(1,2)}$ and the reflection matrices $r_{(1,2)}$:
\begin{equation}\label{t_r_def}
	F^+_+t_{(1)}=F^-_++F^-_-r_{(1)},\qquad F^-_-t_{(2)}=F^+_-+F^+_+r_{(2)}.
\end{equation}
Combining the relations \eqref{basis_expan} and comparing the result with \eqref{t_r_def}, we arrive at
\begin{equation}\label{t1_r1_t2_r2_def}
	t_{(1)}=\Phi_+^{-1},\qquad r_{(1)}=\Psi_+\Phi_+^{-1},\qquad t_{(2)}=\Phi_- -\Psi_+\Phi_+^{-1}\Psi_-,\qquad
    r_{(2)}=-\Phi_+^{-1}\Psi_-.
\end{equation}
The relations \eqref{PhiPsi_rels} imply the symmetry properties
\begin{equation}\label{symmetry_rels}
	K_-t_{(2)}=t_{(1)}^T K_+,\qquad K_-r_{(1)}=r_{(1)}^TK_-,\qquad K_+r_{(2)}=r_{(2)}^TK_+.
\end{equation}
Define the $S$-matrix as
\begin{equation}
	S:=
	\left[
	\begin{array}{cc}
		t_{(1)} & r_{(2)} \\
		r_{(1)} & t_{(2)} \\
	\end{array}
	\right].
\end{equation}
To shorten the notation, we also introduce the operation that acts on the products of matrices $\Phi_\pm$, $\Psi_\pm$ and their inverses by the rule
\begin{equation}
	\bar{\Phi}_\pm:=\Phi_\mp,\qquad \bar{\Psi}_\pm:=\Psi_\mp,\qquad \overline{A^{-1}}=(\bar{A})^{-1},\qquad
    \overline{AB}:=\bar{A}\bar{B}.
\end{equation}
Then the $S$-matrix possesses the symmetries
\begin{equation}\label{S_symmetries}
	\begin{gathered}
		\left[
		\begin{array}{cc}
			0 & K_- \\
			K_+ & 0 \\
			\end{array}
			\right] S=
			S^T
			\left[
			\begin{array}{cc}
			0 & K_+ \\
			K_- & 0 \\
			\end{array}
			\right],\qquad
			\left[
			\begin{array}{cc}
			0 & K_- \\
			K_+ & 0 \\
			\end{array}
			\right] \bar{S}=
			\bar{S}^T
			\left[
			\begin{array}{cc}
			0 & K_+ \\
			K_- & 0 \\
			\end{array}
			\right],\\
			\bar{S}^T
			\left[
			\begin{array}{cc}
			K_+ & 0 \\
			0 & K_- \\
			\end{array}
			\right]S=
			\left[
			\begin{array}{cc}
			K_- & 0 \\
			0 & K_+ \\
		\end{array}
		\right].
	\end{gathered}
\end{equation}
\begin{theorem}\label{ooUnitTh}
	If $\la\in \mathbb{R}$ belongs to none of the cuts of the functions $(K_\pm)_s$, $s=\overline{1,N}$, i.e., when all the scattering channels are open, the $S$-matrix is unitary
	\begin{equation}\label{S_unitar_allopen}
		S^\dag
		\left[
		\begin{array}{cc}
		K_+ & 0 \\
		0 & K_- \\
		\end{array}
		\right]S=
		\left[
		\begin{array}{cc}
		K_- & 0 \\
		0 & K_+ \\
		\end{array}
		\right].
	\end{equation}
\end{theorem}
\begin{proof}
	The proof of the theorem follows directly from the last equality in \eqref{S_symmetries} and the property \eqref{PhiPsi_cc_regul}.
\end{proof}
\begin{remark}
	Introducing the notation
	\begin{equation}
		\tilde{\Phi}_\pm:=K_-^{1/2} \Phi_\pm K_+^{-1/2},\qquad \tilde{\Psi}_\pm:=K_-^{1/2} \Psi_\pm K_+^{-1/2},
	\end{equation}
	one can reduce \eqref{S_unitar_allopen} to the standard form
	\begin{equation}
		\tilde{S}^\dag \tilde{S}=1,
	\end{equation}
	where $\tilde{S}$ is expressed in terms of $\tilde{\Phi}_\pm$ and $\tilde{\Psi}_\pm$ in the same way that $S$ is expressed in terms of $\Phi_\pm$ and $ \Psi_\pm$.
	
\end{remark}

\section{Identities in the subspace of open channels}\label{Identities_Open_Chan_Sec}

It is more difficult to prove unitarity of the $S$-matrix in the case when some of the channels are closed for $z\rightarrow-\infty$ and/or $z\rightarrow\infty$, i.e., when for some fixed $\la\in \mathbb{R}$ some of $(K_\pm)_s$ are purely imaginary. Let the number of open channels for $z\rightarrow-\infty$ be $l_o$, and the number of closed channels be $l_c$. As for the numbers of open and closed channels for $z\rightarrow\infty$, we introduce the notation $r_o$ and $r_c$, respectively. For definiteness, we assume that $l_o\geqslant r_o$. It is clear that $l_o+l_c=r_o+r_c=N$. We split the relations \eqref{t_r_def} into blocks with respect to the indices $s$, $s'$ in accordance with the splitting into open and closed channels,
\begin{subequations}\label{t_r_blocks}
	\begin{align}
	(F^+_+)_ot_{(1)oo}+(F^+_+)_ct_{(1)co}&=(F^-_+)_o+(F^-_-)_or_{(1)oo}+(F^-_-)_cr_{(1)co},\label{7eq}\\
	(F^+_+)_ot_{(1)oc}+(F^+_+)_ct_{(1)cc}&=(F^-_+)_c+(F^-_-)_or_{(1)oc}+(F^-_-)_cr_{(1)cc},\label{8eq}\\
	(F^-_-)_ot_{(2)oo}+(F^-_-)_ct_{(2)co}&=(F^+_-)_o+(F^+_+)_or_{(2)oo}+(F^+_+)_cr_{(2)co},\\
	(F^-_-)_ot_{(2)oc}+(F^-_-)_ct_{(2)cc}&=(F^+_-)_c+(F^+_+)_or_{(2)oc}+(F^+_+)_cr_{(2)cc},\label{10eq}
	\end{align}
\end{subequations}
where, for example,
\begin{equation}
	t_{(1)}=
	\left[
	\begin{array}{cc}
		t_{(1)oo} & t_{(1)oc} \\
		t_{(1)co} & t_{(1)cc} \\
		\end{array}
		\right],\qquad
		F^+_\pm=\left[
		\begin{array}{cc}
		(F^+_\pm)_o & (F^\pm_\pm)_c \\
	\end{array}
	\right].
\end{equation}
Taking complex conjugate of these equations, bearing in mind the above conditions on $\la$, and using the expressions for the Jost solutions \eqref{Jost_sols}, we arrive at
\begin{subequations}
	\begin{align}
	(F^+_-)_ot^*_{(1)oo}+(F^+_+)_ct^*_{(1)co}&=(F^-_-)_o+(F^-_+)_or^*_{(1)oo}+(F^-_-)_cr^*_{(1)co},\label{7*eq}\\
	(F^+_-)_ot^*_{(1)oc}+(F^+_+)_ct^*_{(1)cc}&=(F^-_+)_c+(F^-_+)_or^*_{(1)oc}+(F^-_-)_cr^*_{(1)cc},\label{8*eq}\\
	(F^-_+)_ot^*_{(2)oo}+(F^-_-)_ct^*_{(2)co}&=(F^+_+)_o+(F^+_-)_or^*_{(2)oo}+(F^+_+)_cr^*_{(2)co},\label{9*eq}\\
	(F^-_+)_ot^*_{(2)oc}+(F^-_-)_ct^*_{(2)cc}&=(F^+_-)_c+(F^+_-)_or^*_{(2)oc}+(F^+_+)_cr^*_{(2)cc}.\label{10*eq}
	\end{align}
\end{subequations}
Introducing the notation,
\begin{equation}
	(K_\pm)_o=:\vk_\pm^o,\qquad (K_\pm)_c=:i\vk_\pm^c,\qquad \vk^{o,c}_\pm>0,
\end{equation}
the symmetry relations \eqref{symmetry_rels} become
\begin{equation}\label{symmetry_rels_blocks}
	\begin{gathered}
		\vk^o_- t_{(2)oo}=(t_{(1)oo})^T \vk^o_+,\quad \vk^o_- t_{(2)oc}=(t_{(1)co})^T i\vk^c_+,\quad i\vk^c_- t_{(2)co}=(t_{(1)oc})^T \vk^o_+,\quad \vk^c_- t_{(2)cc}=(t_{(1)cc})^T \vk^c_+,\\
		\vk^o_- r_{(1)oo}=(r_{(1)oo})^T \vk^o_-,\quad \vk^o_- r_{(1)oc}=(r_{(1)co})^T i\vk^c_-,\quad i\vk^c_- r_{(1)co}=(r_{(1)oc})^T \vk^o_-,\quad \vk^c_- r_{(1)cc}=(r_{(1)cc})^T \vk^c_-,\\
		\vk^o_+ r_{(2)oo}=(r_{(2)oo})^T \vk^o_+,\quad \vk^o_+ r_{(2)oc}=(r_{(2)co})^T i\vk^c_+,\quad i\vk^c_+ r_{(2)co}=(r_{(2)oc})^T \vk^o_+,\quad \vk^c_+ r_{(2)cc}=(r_{(2)cc})^T \vk^c_+.
	\end{gathered}
\end{equation}
Without loss of generality, we can assume that the rank of the matrix $t_{(1)oo}$ is maximal and is equal to $r_o$. Then it follows from the first relation in \eqref{symmetry_rels_blocks} that the rank of the matrix $t_{(2)oo}$ is also $r_o$.

In this case we have
\begin{equation}
	t_{(1)oo}t^\vee_{(1)oo}=t^\vee_{(2)oo}t_{(2)oo}=1,\qquad t^\vee_{(1)oo}t_{(1)oo}=:L_{(1)},\qquad t_{(2)oo}t^\vee_{(2)oo}=:L_{(2)},
\end{equation}
where $A^\vee$ is a pseudo-inverse matrix to $A$, and $L_{(1)}$ and $L_{(2)}$ are Hermitian projectors of rank $r_o$ in the subspace of open channels at $z\rightarrow-\infty$. It follows from the definition of pseudo-inverse matrix that
\begin{equation}\label{tbarL}
	t_{(1)oo}\bar{L}_{(1)}=\bar{L}_{(1)}t^\vee_{(1)oo}=0,\qquad \bar{L}_{(2)}t_{(2)oo}=t^\vee_{(2)oo}\bar{L}_{(2)}=0,
\end{equation}
where $\bar{L}_{(1,2)}:=1-L_{(1,2)}$. Further, we express the functions $ (F^+_-)_o$ from \eqref{7*eq} and substitute them into \eqref{9*eq}. This gives rise to
\begin{multline}
	(F^+_+)_o +(F^+_+)_c[r^*_{(2)co} -t^*_{(1)co}(t^*_{(1)oo})^\vee r^*_{(2)oo}]=\\
	=(F^-_+)_o[t^*_{(2)oo}-r^*_{(1)oo}(t^*_{(1)oo})^\vee r^*_{(2)oo}] -(F^-_-)_o(t^*_{(1)oo})^\vee r^*_{(2)oo} +(F^-_-)_c[t^*_{(2)co}-r^*_{(1)co}(t^*_{(1)oo})^\vee r^*_{(2)oo}].
\end{multline}
Then multiplying \eqref{7eq} by $t_{(1)oo}^\vee$ and comparing the result with the equation above, we have
\begin{subequations}
	\begin{align}
		t_{(1)oo}^\vee&=t^*_{(2)oo}-r^*_{(1)oo}(t^*_{(1)oo})^\vee r^*_{(2)oo},\label{tv1oo}\\
		r_{(1)oo}L_{(1)}&=-(t^*_{(1)oo})^\vee r^*_{(2)oo} t_{(1)oo},\label{r1oo}\\
		t_{(1)co}L_{(1)}&=[r^*_{(2)co}-t^*_{(1)co}(t^*_{(1)oo})^\vee r^*_{(2)oo}] t_{(1)oo},\\
		r_{(1)co}L_{(1)}&=[t^*_{(2)co}-r^*_{(1)co}(t^*_{(1)oo})^\vee r^*_{(2)oo}] t_{(1)oo}.
	\end{align}
\end{subequations}
We infer from \eqref{r1oo} that
\begin{equation}\label{barLrL}
	\bar{L}^*_{(1)}r_{(1)oo}L_{(1)}=0.
\end{equation}
Then \eqref{tv1oo} implies
\begin{equation}
	\bar{L}_{(1)}t^*_{(2)oo}=0.
\end{equation}
Consequently,
\begin{equation}
	\bar{L}_{(1)}L^*_{(2)}=0\;\Rightarrow\;L_{(1)}L^*_{(2)}=L^*_{(2)}=L^*_{(2)}L_{(1)}.
\end{equation}
As the ranks of $L_{(1)}$ and $L^*_{(2)}$ are the same, we have
\begin{equation}
	L_{(1)}=L^*_{(2)}=:L.
\end{equation}
The first relation in \eqref{symmetry_rels_blocks} implies
\begin{equation}
	t^\vee_{(2)oo}=(\vk_+^o)^{-1}(t^T_{(1)oo})^\vee\vk_-^o L^*,
\end{equation}
whence
\begin{equation}
	L^*=t_{(2)oo}t^\vee_{(2)oo}=(\vk_+^o)^{-1}L^T \vk_+^oL^*.
\end{equation}
Therefore,
\begin{equation}
	\vk_+^oL =L\vk_+^oL= L\vk_+^o.
\end{equation}
Thus we see that $L$ is a diagonal matrix and hence $L^*=L$.

\section{Unitarity in the subspace of open channels}\label{Unitarity_Sec}

Now we are in position to prove unitarity of the $S$-matrix in open channels as well as to obtain the other relations involving the different components of $t_{(1,2)}$ and $r_{(1,2)}$.
\begin{theorem}\label{ccUnitTh}
	The $S$-matrix in the subspace of open channels is unitary.
\end{theorem}
\begin{proof}
	It follows from \eqref{r1oo} that
	\begin{equation}\label{unit1}
		t^*_{(1)oo} r_{(1)oo}L+r^*_{(2)oo}t_{(1)oo}=0\;\Rightarrow\; t^*_{(1)oo} r_{(1)oo}+r^*_{(2)oo}t_{(1)oo}=0,
	\end{equation}
	where the properties \eqref{tbarL}, \eqref{barLrL} have been taken into account in the last equality. Given the relation \eqref{unit1}, equality \eqref{tv1oo} implies
	\begin{equation}\label{unit2}
		t^*_{(2)oo}t_{(1)oo}+r^*_{(1)oo}r_{(1)oo}L=L.
	\end{equation}
	Using the symmetry relations \eqref{symmetry_rels_blocks} in \eqref{unit1}, we obtain
	\begin{equation}\label{unit1p}
		r_{(1)oo}t^*_{(2)oo}+t_{(2)oo}r^*_{(2)oo}=0.
	\end{equation}
	Taking complex conjugate of \eqref{unit2} and multiplying the resulting expression by $t^\vee_{(1)oo}$ from the left and by $t_{(1)oo}$ from the right, using \eqref{unit1}, \eqref{unit1p}, we come to
	\begin{equation}\label{unit2p}
		t^*_{(1)oo}t_{(2)oo}+r^*_{(2)oo}r_{(2)oo}=1.
	\end{equation}
	The relations \eqref{unit1}, \eqref{unit2}, \eqref{unit1p}, \eqref{unit2p} lead to the unitarity relations for the $S$-matrix in open channels. Indeed, substituting the symmetry relations \eqref{symmetry_rels_blocks} into these expressions, we have
	\begin{equation}\label{unit_rels}
		\begin{split}
			t^\dag_{(2)oo}\vk^o_-r_{(1)oo}L+r^\dag_{(2)oo}\vk^o_+t_{(1)oo}&=0,\\
			t^\dag_{(1)oo}\vk^o_+ t_{(1)oo} +r^\dag_{(1)oo}\vk^o_-r_{(1)oo}L&=\vk^o_-L,\\
			r^\dag_{(1)oo}\vk^o_-t_{(2)oo}+t^\dag_{(1)oo}\vk^o_+r_{(2)oo}&=0,\\
			t^\dag_{(2)oo}\vk^o_-t_{(2)oo} +r^\dag_{(2)oo} \vk^o_+r_{(2)oo}&=1,
		\end{split}
	\end{equation}
	where the complex conjugate equation \eqref{unit1p} has been used in the third equality.
	
	There are also the relations for the matrix $r_{(1)oo}$. In order to deduce them, multiply \eqref{7*eq} by $\bar{L}r_{(1)oo}\bar{L}$ from the right and compare with \eqref{7eq} multiplied by $\bar{L}$ from the right. As a result, we have
	\begin{equation}
		t_{(1)co}\bar{L}=t^*_{(1)co}\bar{L}r_{(1)oo}\bar{L},\qquad r_{(1)co}\bar{L}=r^*_{(1)co}\bar{L}r_{(1)oo}\bar{L},\qquad \bar{L}=r^*_{(1)oo}\bar{L}r_{(1)oo}\bar{L}.
	\end{equation}
	Combining the last relation with \eqref{unit2}, we find
	\begin{equation}
		t^*_{(2)oo}t_{(1)oo}+r^*_{(1)oo}r_{(1)oo}=1.
	\end{equation}
	As for the second relation in \eqref{unit_rels}, it is written as
	\begin{equation}\label{unit_rels_p}
		t^\dag_{(1)oo}\vk^o_+ t_{(1)oo} +r^\dag_{(1)oo}\vk^o_-r_{(1)oo}=\vk^o_-.
	\end{equation}
	The relations \eqref{unit_rels}, \eqref{unit_rels_p} are nothing but the unitarity relations for the $S$-matrix \eqref{S_unitar_allopen} in the subspace of open channels.
\end{proof}

There are also the additional relations connecting the components $t_{(1,2)}$ and $r_{(1,2)}$ of closed and open channels.
\begin{proposition}
	The following relations hold
	\begin{equation}
		\begin{split}
			t_{(1)cc}&=t^*_{(1)cc}-r_{(2)co}t^*_{(1)oc}-t_{(1)co}r^*_{(1)oc},\\
			r_{(1)cc}&=r^*_{(1)cc}-t_{(2)co}t^*_{(1)oc} -r_{(1)co}r^*_{(1)oc},\\
			t_{(1)co}&=t^*_{(1)co}r_{(1)oo}+r^*_{(2)co}t_{(1)oo},\\
			r_{(1)co}&=t^*_{(2)co}t_{(1)oo}+r^*_{(1)co}r_{(1)oo},\\
			t_{(1)oc}&=-r_{(2)oo}t^*_{(1)oc}-t_{(1)oo}r^*_{(1)oc},\\
			r_{(1)oc}&=-t_{(2)oo}t^*_{(1)oc}-r_{(1)oo}r^*_{(1)oc}.
		\end{split}
	\end{equation}
	\begin{remark}
		There are also the relations obtained from these ones by replacing $1\leftrightarrow2$.
	\end{remark}
\end{proposition}
\begin{proof}
	To deduce these relations, we express $(F^-_+)_o$ and $(F^+_-)_o$ from \eqref{7eq} times $\bar{L}$, \eqref{7*eq} times $L$, and \eqref{9*eq}. Then we substitute these Jost solutions into \eqref{10*eq} and compare the result with \eqref{10eq}. Further, we also substitute $(F^-_+)_o$ and $(F^+_-)_o$ found in this way into \eqref{8*eq} and compare the result with \eqref{8eq}. Having carried out this, we arrive at the relations from the statement of the proposition.
\end{proof}

\section{Bound states}\label{Bouns_States_Sec}

Let $\det\Phi_+(\la)=0$ for some $\la\in \mathbb{C}$. Then
\begin{equation}\label{Phi_zeroeigen}
	\exists v(\la)\neq0, w(\la)\neq0:\qquad\Phi_+(\la)v(\la)=0,\qquad w^T(\la)\Phi_+(\la)=0.
\end{equation}
In a general position, the rank of $\Phi_+(\la)$ drops by one at the given point $\la$. Therefore, we can assume that the conditions \eqref{Phi_zeroeigen} determine the vectors $v$ and $w$ uniquely up to multiplication by a constant. It follows from the first relation in \eqref{PhiPsi_rels} that
\begin{equation}
	\Phi_+^T K_-\Psi_+v=0\;\Rightarrow\;\Psi_+v=K_-^{-1}w.
\end{equation}
By multiplying the first relation in \eqref{basis_expan} by $v$ from the right, we obtain a particular solution to Eq. \eqref{Schr_eqn} of the form
\begin{equation}\label{bound_sol}
	F^+_+v=F^-_-\Psi_+v=F^-_-K_-^{-1}w.
\end{equation}
Given the asymptotic behavior of Jost solutions \eqref{Jost_asympt}, this solution decreases exponentially for $z\rightarrow\pm\infty$, provided that $\arg\sqrt{\La_s^\pm-\la}\in(0,\pi)$ for all $s$, whereas it increases exponentially for $z\rightarrow\pm\infty$ provided that $\arg\sqrt{\La_s^\pm-\la}\in(-\pi,0)$ for all $s$. Note that for the branch of the root we have chosen, the values of the argument  $(-\pi,-\pi/2)\cup(\pi/2,\pi)$ correspond to the second sheet of the Riemann surface. Since Eq. \eqref{Schr_eqn} is a spectral problem for a self-adjoint operator, this equation cannot possess square-integrable solutions for $\la\not\in\mathbb{R}$. Hence $\det\Phi_+(\la)\neq0$ when $\la\not\in\mathbb{R}$ and $\arg\sqrt{\La_s^\pm-\la}\in(0 ,\pi)$ for all $s$.

In virtue of the restrictions \eqref{finite_supp} on the asymptotic behavior of $g_{ij}(z)$ and $V_{ij}(z)$, the bound solutions to Eq. \eqref{Schr_eqn} tend exponentially to zero when $z\rightarrow\pm\infty$. Therefore, as the Jost solutions $F^+_\pm$ and $F^-_\pm$  constitute bases in the space of solutions to Eq. \eqref{Schr_eqn}, the condition \eqref{Phi_zeroeigen} is a necessary and sufficient condition for the existence of a bound state provided $\la\in \mathbb{R}$ and $v_o=0$ and $w_o=0$. The corresponding bound state is given by formula \eqref{bound_sol}.

The following obvious assertion is valid.
\begin{proposition}
	If all scattering channels are open, there are no bound states.
\end{proposition}
\begin{proof}
	Let $\la\in \mathbb{R}$ be such that all the scattering channels are open, i.e., all $\sqrt{\La_s^\pm-\la}$ are real. Then, multiplying the second relation in \eqref{PhiPsi_rels} by $v^T$ from the left and by $v^*$ from the right, we arrive at
	\begin{equation}
		-v^T\Psi_+^TK_-\Psi_-v^*=-v^T\Psi_+^TK_-(\Psi_+v)^*=v^TK_+v^*.
	\end{equation}
	The expression on the left-hand side is negative-definite, while the expression on the right-hand side is positive-definite. Therefore, $v=0$.
\end{proof}

Now let $\la\in \mathbb{R}$ be such that some of the scattering channels are closed, as described before formulas \eqref{t_r_blocks}, and the condition \eqref{Phi_zeroeigen} is satisfied. Then
\begin{proposition}
	 The particular solution \eqref{bound_sol} includes only those Jost solutions that correspond to closed channels, i.e., $v_o=0$ and $w_o=0$. This is a bound state.
\end{proposition}
\begin{proof}
	Partition the matrix $\Phi_+$ into blocks
	\begin{equation}
		\Phi_+=
		\left[
		\begin{array}{cc}
			(\Phi_+)_{11} & (\Phi_+)_{12} \\
			(\Phi_+)_{21} & (\Phi_+)_{22} \\
		\end{array}
		\right],
	\end{equation}
	where the block $(\Phi_+)_{11}$ has dimensions $r_o\times r_o$ and acts from the open channels on the right to the open channel subspace on the left, this subspace being distinguished by the projector $L$. Then there are the relations
	\begin{equation}\label{v_w_rels}
		\begin{aligned}
			(\Phi_+)_{11}v_1+(\Phi_+)_{12}v_2&=0,&\qquad (\Phi_+)_{21}v_1+(\Phi_+)_{22}v_2&=0,\\
			w_1^T(\Phi_+)_{11}+w_2^T(\Phi_+)_{21}&=0,&\qquad w_1^T(\Phi_+)_{12}+w_2^T(\Phi_+)_{22}&=0,
		\end{aligned}
	\end{equation}
	and
	\begin{equation}
		t^{-1}_{(1)11}=(\Phi_+)_{11}-(\Phi_+)_{12}(\Phi_+)^{-1}_{22}(\Phi_+)_{21}.
	\end{equation}
	Acting on the last expression by $w_1^T$ and $v_1$ from the left and from the right, respectively, and employing the relations \eqref{v_w_rels}, we see that $w_1$ and $v_1$ are the left and right null vectors of $t^{- 1}_{(1)11}$. However, the unitarity relation \eqref{unit_rels_p} implies that $t_{(1)11}=t_{(1)oo}L$ is a bounded operator. Hence $w_1=Lw_o=0$ and $v_1=v_o=0$.
	
	It follows from the first relation in \eqref{t1_r1_t2_r2_def} that
	\begin{equation}
		(\Phi_+)_{oo}t_{(1)oo}+(\Phi_+)_{oc}t_{(1)co}=1,\qquad (\Phi_+)_{co}t_{(1)oo}+(\Phi_+)_{cc}t_{(1)co}=0.
	\end{equation}
	Whence, multiplying both the equalities by $\bar{L}$ from the right, we obtain
	\begin{equation}\label{aux_rels}
		(\Phi_+)_{oc}t_{(1)co}\bar{L}=\bar{L},\qquad (\Phi_+)_{cc}t_{(1)co}\bar{L}=0.
	\end{equation}
	Since
	\begin{equation}\label{w_Phi}
		w_o^T(\Phi_+)_{oo}+w_c^T(\Phi_+)_{co}=0,\qquad w_o^T(\Phi_+)_{oc}+w_c^T(\Phi_+)_{cc}=0,
	\end{equation}
	multiplying the first relation in \eqref{aux_rels} by $w_o^T$ from the left and using the second relation in \eqref{w_Phi}, we arrive at
	\begin{equation}
		w_o^T\bar{L}=-w_c^T(\Phi_+)_{cc}t_{(1)co}\bar{L}=0,
	\end{equation}
	where the last equality follows from the second relation in \eqref{aux_rels}. Thus, we have proved that
	\begin{equation}
		v_o=0,\qquad w_o=0.
	\end{equation}
	As a result, both the left- and right-hand sides of equality \eqref{bound_sol} defining a particular solution to Eq. \eqref{Schr_eqn} include only those Jost solutions that correspond to the closed scattering channels. Therefore, this particular solution decreases exponentially as $z\rightarrow\pm\infty$ and so it is a bound state. Notice that $v_c\neq0$ and $w_c\neq0$ since otherwise $v=0$ and $w=0$ in contradiction with \eqref{Phi_zeroeigen}.
\end{proof}

Thus we have proved the theorem
\begin{theorem}
	The following condition
	\begin{equation}
		\det\Phi_+(\la)=0,\quad\la\in \mathbb{R},
	\end{equation}
	is a necessary and sufficient condition for the existence of bound states of Eq. \eqref{Schr_eqn}.
\end{theorem}

\section{Shortwave asymptotics}\label{Shartwave_Asypt_Sec}

Notice that
\begin{equation}
	\det S(\la)=\frac{\det\Phi_-(\la)}{\det\Phi_+(\la)}=\frac{\det\tilde{\Phi}_-(\la)}{\det\tilde{\Phi}_+(\la)}=\det\tilde{S}(\la).
\end{equation}
Let us show that $\det\Phi_\pm(\la)\rightarrow1$ and $\det\tilde{\Phi}_\pm(\la)\rightarrow1$ as $|\la|\rightarrow\infty$ and so $\det S(\la)\rightarrow1$ and $\det \tilde{S}(\la)\rightarrow1$ in this limit. Let us find the explicit expressions for the Jost solutions
\begin{equation}
	\tilde{F}^+_\pm:=F^+_\pm K_+^{-1/2},\qquad \tilde{F}^-_\pm:=F^-_\pm K_-^{-1/2},
\end{equation}
when $|\la|\rightarrow\infty$. In this limit, one can use the semiclassical matrix approximation \cite{BabBuld,Maslov,BagTrif,RMKD18} to obtain a solution of Eq. \eqref{Schr_eqn}. We are looking for a solution of Eq. \eqref{Schr_eqn} in the form
\begin{equation}
	f_{is}(z)e^{iS_s(z)}\qquad (\text{no summation over $s$}).
\end{equation}
Then, in leading order, we have \cite{BabBuld,BKKL21}
\begin{equation}
	g_{ij}(z)f_{js}(z)\La_s(z)=V_{ij}(z)f_{js}(z)\qquad (\text{no summation over $s$}),
\end{equation}
where $\La_s(z)\in \mathbb{R}$ are eigenvalues and
\begin{equation}
	K_s(z)=S'_s(z)=\sqrt{\La_s(z)-\la}, \qquad f^T(z)g(z)f(z)=1.
\end{equation}
The following relations are also fulfilled
\begin{equation}
	[K_s(z) f^\dag_s(z)g(z)f_s(z)]'=0,\qquad\im[f^\dag_s(z)g(z)f'_s(z)]=0.
\end{equation}
As a result, we obtain the semiclassical expressions for the Jost solutions
\begin{equation}\label{Jost_semiclas}
	(\tilde{F}^+_\pm)_{is}(z)=\frac{f_{is}(z)}{K^{1/2}_s(z)}e^{\pm iS^+_s(z)},\qquad (\tilde{F}^-_\pm)_{is}(z)=\frac{f_{is}(z)}{K^{1/2}_s(z)}e^{\pm iS^-_s(z)},
\end{equation}
where
\begin{equation}
	S^+_s(z)=(K_+)_sL_z-\int_z^{L_z} dz'K_s(z'),\qquad S^-_s(z)=\int_{-L_z}^z dz'K_s(z')-(K_-)_sL_z,
\end{equation}
and $L_z$ was defined in \eqref{finite_supp}. One can replace $L_z$ in the expressions for $S_s^\pm(z)$ by any real number greater than or equal to $L_z$.

Using the first relation in \eqref{PhiPsiWronsk}, we derive
\begin{equation}
	2i(\tilde{\Phi}_+)_{ss'}=w[(\tilde{F}^-_-)_s,(\tilde{F}^+_+)_{s'}]=2i\de_{ss'}e^{i(S_s^+-S_s^-)}.
\end{equation}
In the last equality, the semiclassical expressions \eqref{Jost_semiclas} for the Jost solutions have been employed and
\begin{equation}
	S^+_s(z)-S^-_s(z)=[(K_+)_s +(K_-)_s]L_z-\int_{-L_z}^{L_z}dz' K_s(z').
\end{equation}
It is evident from the explicit expressions for $K_\pm(\la)$ and $K(\la)$ that
\begin{equation}
	S^+_s(z)-S^-_s(z)\underset{|\la|\rightarrow\infty}{\rightarrow}0.
\end{equation}
Therefore,
\begin{equation}
	(\tilde{\Phi}_+)_{ss'}\underset{|\la|\rightarrow\infty}{\rightarrow}\de_{ss'},
\end{equation}
and
\begin{equation}\label{detPhi}
	\det\tilde{\Phi}_+(\la)\underset{|\la|\rightarrow\infty}{\rightarrow}1.
\end{equation}
As long as $\tilde{\Phi}_-(\la)=\tilde{\Phi}^*_+(\la^*)$ outside the cuts, the asymptotics \eqref{detPhi} is also valid for $\det\tilde {\Phi}_-(\la)$. It is clear that the same asymptotics hold for $\Phi_\pm(\la)$.

\section{Conclusion}

Let us summarize the results. We have considered a multichannel stationary scattering problem for the one-dimensional Schr\"{o}dinger equation with a potential $V_{ij}(z)$ and an inverse mass matrix $g_{ij}(z)$, where $z\in \mathbb{R}$. The matrices $V_{ij}(z)$ and $g_{ij}(z)$ are assumed to be real, symmetric, and equal to constant values for sufficiently large $|z|$. Their matrix elements are supposed to be piecewise continuous functions. The matrix $g_{ij}(z)$ is assumed to be positive-definite. We consider the general case and do not suppose that the asymptotics of the matrices $V_{ij}(z)$ and $g_{ij}(z)$ for $z\rightarrow-\infty$ and for $z\rightarrow\infty$ coincide.

The analytical structure of the Jost solutions $(F^+_\pm)_{is}$ and $(F^-_\pm)_{is}$ as functions of the auxiliary spectral parameter $\la$ has been investigated. It has been shown that the Jost solutions $(F^+_\pm)_{is}$ are the different branches of the same vector-valued analytic function on a double-sheeted Riemann surface. The same is true for $(F^-_\pm)_{is}$ but with different branch points. It has also been shown that the matrix-valued functions that form the transition matrix between the bases $(F^+_\pm)_{is}$ and $(F^-_\pm)_{is}$ are the different branches of the same analytic matrix-valued function.

The multichannel scattering matrix has been investigated. The key result of this paper is the Theorem \ref{ccUnitTh} that proves  unitarity of the scattering matrix in the subspace of open channels. One would expect that the scattering matrix should be unitary in the subspace of open channels on physical grounds. It is also clear that in the general case the complete $S$-matrix is not unitary in the presence of closed channels. Nevertheless, the proof of this fact is nontrivial and appears to be obtained in this paper for the first time. A weak version of this theorem, viz., the statement about unitarity of the $S$-matrix in the case when all the scattering channels are open, has also been proved (Theorem \ref{ooUnitTh}). The proof of the latter statement is known in the literature in the case when the asymptotics of the matrices $g_{ij}(z)$ and $V_{ij}(x)$ as $z\rightarrow\pm\infty$ coincide \cite{Ning1995,vanDijk2008,Calogero1976,Zakhariev1990,Kiers1996,Aktosun2001,Bondarenko2017}. In addition to the unitarity relations for the $S$-matrix in the subspace of open scattering channels, the other relations connecting the components of the reflection and transmission matrices in the subspaces of closed and open channels have been deduced.

The condition determining the bound states has been obtained. In particular, it has been shown that the necessary condition for the presence of bound states is the presence of a nonzero subspace of closed channels. The asymptotics of the Jost solutions and of the transition matrix at a large spectral parameter have been investigated. It has been shown that the Schr\"{o}dinger equation under study is solvable in the shortwave approximation. The explicit expressions for the Jost solutions in the semiclassical approximation and the asymptotics of the transition matrix between the bases constituted by the Jost solutions have been obtained.

The results of the paper are applicable in electrodynamics of continuous media \cite{BelyakovBook,KazKor22,BKKL23}, in sound wave propagation theory \cite{Oldano2000,Mabuza2005,Georgiannis2011}, in describing the passage of electrons through heterostructures \cite{Botha2008,Sofianos2007,Liu1996,Shubin2019,Miroshnichenko2010}, in quantum chemistry \cite{Friedman1999}, in hydrodynamics and plasma physics \cite{Adam1986}, etc. In particular, the issue of proving unitarity of the $S$-matrix in open channels arises in describing photon scattering by metamaterials with a large spatial dispersion. The presence of spatial dispersion is caused by the presence of additional degrees of freedom -- the plasmon polaritons, which exist only inside the medium. As a result, there are always the closed channels in scattering of photons by such media, and unitarity becomes less obvious from the physical point of view. In this paper, we prove that the $S$-matrix is unitary for such systems as well provided the appropriate boundary conditions on the additional degrees of freedom are imposed.

\paragraph{Acknowledgments.}

This study was supported by the Tomsk State University Development Programme (Priority-2030).

%\newpage%\selectlanguage{english}

\end{document}